\documentclass[reprint, aps, prl, superscriptaddress, nofootinbib, 10pt]{revtex4-2}

\usepackage[utf8]{inputenc}
\usepackage{amsmath, amssymb, amsfonts, amsthm, amsopn}
\usepackage{dsfont}
\usepackage{thm-restate}
\usepackage{mathtools}

\usepackage{graphicx}
\usepackage{booktabs}

\usepackage{xcolor}

\usepackage[ruled,vlined,linesnumbered]{algorithm2e}

\usepackage[breaklinks]{hyperref}
\hypersetup{
    colorlinks,
    linkcolor={red!40!black},
    citecolor={blue!50!black},
    urlcolor={blue!20!black}
}

\newtheorem{theorem}{Theorem}
\newtheorem*{theorem*}{Theorem}

\newtheorem{lemma}[theorem]{Lemma}
\newtheorem{definition}[theorem]{Definition}

\newtheorem{observation}[theorem]{Observation}

\newtheorem*{problem*}{Problem}
\newtheorem*{question*}{Question}

\newtheorem*{result*}{Result}

\DeclareMathOperator{\tr}{tr}

\newcommand{\bra}[1]{\mathinner{\langle #1|}}
\newcommand{\ket}[1]{\mathinner{|#1\rangle}}

\renewcommand{\a}{\alpha}
\renewcommand{\b}{\beta}
\renewcommand{\c}{\gamma}

\newcommand{\R}{\mathds{R}}

\newcommand{\expect}[1]{\langle #1 \rangle}

\newcommand{\nn}{\nonumber}

\begin{document}
	\title
	[A 0.651-approximation to quantum Max Cut via Rydberg atoms]
	{A 0.651-approximation to quantum Max Cut via Rydberg atoms}

	\date{\today}
	\author{Tomás Crosta}
	\affiliation{Université de Bordeaux, CNRS, LaBRI, 351 cours de la Libération, 33405 Talence, France}
	\affiliation{Université de Bordeaux, CNRS, LOMA, UMR 5798, F-33400 Talence, France}
	\email{tomycrosta@gmail.com}
	\author{Matthieu Saubanere}
	\affiliation{Université de Bordeaux, CNRS, LOMA, UMR 5798, F-33400 Talence, France}
	\author{Felix Huber}
	\affiliation{
		Division of Quantum Information,
		Faculty of Mathematics, Physics and Informatics,
		University of Gdańsk,
		Wita Stwosza 57, 80-308 Gdańsk, Poland
	}
	\email{felix.huber@ug.edu.pl}

	\thanks{
	We thank
	Juan Mauricio Matera,
	Ojas Parekh,
	Gerard Valentí-Rojas,
	and
	Adrian Tanasa,
	for fruitful discussions.
	This research was funded in whole or in part by the
	National Science Centre, Poland 2024/54/E/ST2/00451
	and
	the Polish National Agency for Academic Exchange under the Strategic Partnership Programme grant
	BNI/PST/2023/1/00013/U/00001.
	For the purpose of Open Access, the author has applied a CC-BY public copyright licence to any Author Accepted Manuscript (AAM) version arising from this submission.
	 This work was supported in part by the Maison du Quantique de Nouvelle-Aquitaine “HybQuant”, as
	part of the HQI initiative and France 2030, under the French National Research Agency (ANR)
	grant ANR-22-PNCQ-0002 and the PEPR Quantique under Grant No. ANR-23-PETQ-0006.
	}

	\begin{abstract}
	Quantum Max Cut, also known as the anti-ferromagnetic Heisenberg Hamiltonian, 
	is a QMA-complete problem which serves as a benchmark for approximation algorithms in quantum physics. Here we develop a {\it hybrid} approximation algorithm to quantum Max Cut, 
		which uses the natural quantum dynamics of Rydberg atom systems
		in combination with semidefinite programming and randomized rounding.
		It achieves a conditional approximation ratio of $0.651$,
		compared to the best-known ratio of $0.614$ that relies on semidefinite programming alone.
		The algorithm is robust in the sense that the advantage persists
		even if the annealing procedure of the Rydberg atom system obtains a state whose energy
		is only $89\%$ of its true ground state energy.
		Our approach opens a new route for hybrid quantum-classical algorithms
		that combine quantum with classical optimization methods.
	\end{abstract}
	\maketitle
	\setcounter{tocdepth}{1}
	
\noindent
{\bf Introduction. --- }
	Approximation algorithms have seen a surge
	in the study of quantum many-body systems and quantum information theory.
	Of strong interest is the problem of {quantum Max Cut} (QMC),
	which is the quantum analogue to the classical Max Cut problem~\cite{maxcut}.
	The QMC problem asks to find the ground state energy of the anti-ferromagnetic
	quantum Heisenberg Hamiltonian\footnote{
	We here consider the task of minimizing a Hamiltonian $H$ (i.e. finding its smallest eigenvalue),
	which corresponds to maximizing $-H$ in the computer science literature.},
	\begin{equation}\label{eq:QMc_Hamiltonian}
	H^{\rm qmc} = -\frac{1}{4}\sum_{(ij)\in E} \omega_{ij} \left(I - X_iX_j-Y_iY_j-Z_iZ_j\right)\,,
\end{equation}
where $E$ is the edge set of some graph.
	Quantum Max Cut is known to be QMA-complete:
	It cannot be solved efficiently even
	when having access to a universal quantum computer~\cite{cubitt2016complexity}.
	Having efficient methods to upper and lower bound
	the ground state energy of $H^{\rm qmc} $ is thus highly desirable.

	The relevance of the Heisenberg model for quantum magnetism makes QMC a
	natural test-bed for the performance of {\it approximation algorithms}.
	An approximation algorithm consists of a polynomial-time method to find an approximate solution
	with some guaranteed quality
	to a (typically) computationally hard problem \cite{helton, Navascués_2008, Pironio_2010, baumgratz2012lower}.
	Given an instance $G$ of a problem,
	these methods often formulate
	the task of finding the optimum $\text{OPT}(G)$
	as a relaxation of a polynomial 
	optimization problem.
	This relaxation then is rounded to a feasible point
	with objective value $\text{ALG}(G)$.
	A guarantee of the algorithm is given by
	an approximation ratio $\alpha \in [0,1]$ such that
	 \begin{equation}
	 \frac{\text{ALG}(G) }{\text{OPT}(G)} \geq \alpha
	 \end{equation}
	holds on all instances of the problem.
	For quantum Max Cut, current approximation algorithms employ ansatz solutions
	that combine one and two qubit states
	derived from semidefinite programming~\cite{ParekhAlmost2021, anshu2020beyond, ParekhApplication2021, lee2022optimizing, king2022improved, lee2024improved, gribling2025improved, apte2025improvedalgorithmsquantummaxcut, bakshi2026sharp}.
	The current best approximation algorithm using this method achieves $\alpha = 0.614$ \cite{apte2025improvedalgorithmsquantummaxcut, apte2025conjecturedbounds2localhamiltonians, bakshi2026sharp}\footnote{
	Given an optimal cut of the graph underlying $H^{\rm qmc}$, there exists an algorithm with a
	conditional approximation ratio of $0.625$ \cite{apte2025conjecturedbounds2localhamiltonians, bakshi2026sharp}.}.
	However, such semidefinite programming relaxations are associated with an adverse computational scaling, making
	alternative approaches come into focus~\cite{huber2024second}.

\begin{figure}[tbp]
	\includegraphics[width=1\linewidth]{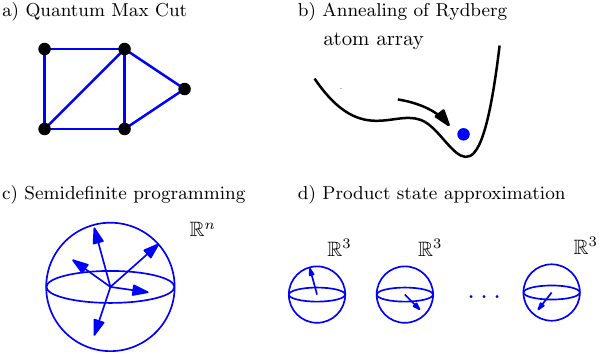}
	\caption{{\bf Sketch of the hybrid algorithm.}
		{\bf a)} A quantum Max Cut instance is encoded into
		{\bf b)} a Hamiltonian of Rydberg atoms with a $(X_i X_j + Y_i Y_j)$
		interaction per edge.
		The Rydberg ground state is prepared via annealing
		and its correlations
		$\langle X_i X_j\rangle$,
		$\langle Y_i Y_j \rangle$
		measured.
		These serve as input to {\bf c)}
		a semidefinite program,
		whose vector solution is rounded via a random projection from $\R^n$ to $\R^3$,
		yielding {\bf d)} a product state.
		The approximation to quantum Max Cut is obtained by returning,
		depending on which performs better,
		either the Rydberg ground state or the product state.
		\label{fig:diagram_algorithm}}
\end{figure}

	Quantum computing promises
	an  efficient approach for solving hard combinatorial optimization problems, for example via annealing or variational algorithms
	\cite{
	lubinski2022advancing,
	cerezo2021variational,
	kim2025quantumannealingcombinatorialoptimization,
	schmid2025hierarchicaldivideconquerquantum,
	zhou2020quantum,
	tene2026variational,
	wang2018quantum,
	Banks2024rapidquantum}.
	Recent experiments suggest that quantum platforms can outperform  approximations from classical computers in some specific tasks
	\cite{
	PhysRevLett.134.090601,
	PhysRevLett.131.150601,
	abanin2025constructive,
	benchmark,
	AmiraChallenges,
	kshetrimayum2026quantumadvantagetensornetwork,
	hangleiter2026quantumadvantageachieved}.
	This raises the question:
	Can classical and quantum methods
	be combined to give stronger approximation guarantees?
	In particular, can one make use of ground states
	of a Hamiltonian that does not encode the original problem,
	but a related one?

	Here we show that this is indeed the case:
	given the ground state of a Rydberg Hamiltonian with
	$XX + YY$ interactions,
	followed by a classical optimization and randomized rounding step,
	we obtain
	an approximation ratio of $0.651$ to quantum Max Cut.
	Our quantum-classical hybrid algorithm relies on the fact
	that the Rydberg Hamiltonian contains two out of the three QMC interaction terms per edge.
	The ground state of this Rydberg system then
	takes two key functions:
	First, it yields an improved
	lower bound on the true ground state energy.
	Second, it serves as a starting point for a classical semidefinite programming machinery,
	from which, via randomized rounding,
	one can construct a product state that
	approximates the true ground state.
	The approximation is given by the
	the Rydberg ground state or the product state
	depending on which achieves a lower energy.
    This hybrid method yields a higher
	approximation ratio than the current classical approaches relying on mathematical optimization alone.
		\begin{figure}[tbp]
		\centering
		\includegraphics[width=0.9\columnwidth]{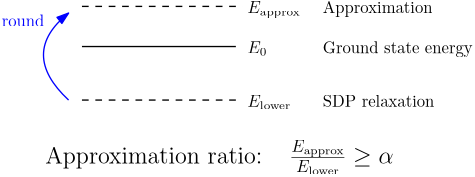}
		\caption{A lower bound of the solution is rounded into a physical state whose energy upper bounds the optimum.
			The approximation ratio of the algorithm then informs about the quality of approximation.}
		\label{fig:illustration_rounding}
	\end{figure}

\smallskip
\noindent
{\bf Moment matrices and randomized rounding. --- }
	Quantum Max Cut (QMC) is the problem of finding the ground state energy
	of the Hamiltonian $H^{\rm qmc}$ constructed from a non-negative weighted graph $G=(V, E, \omega)$, whose edges correspond to Heisenberg interactions,
	\begin{equation}
		H^{\rm qmc} = -\frac{1}{4}\sum_{(ij)\in E} \omega_{ij} \left(I - X_iX_j-Y_iY_j-Z_iZ_j\right)\,.
	\end{equation}
	Here $X_i, Y_i, Z_i$ denote the tensor product operators acting
	on system $i$ with the Pauli matrices $X,Y,Z$
	and with identity on the remaining tensor factors. 
	Technically, finding the largest eigenvalue of $-H^{\rm qmc}$ is known as quantum Max Cut; however in the physics literature one often considers the more natural ground state energy problem.
	
	A converging sequence of lower bounds to this problem is provided by the moment (or Navascués-Pironio-Acín) hierarchy, which at every level consists of a semidefinite programming relaxation to the original problem~\cite{helton, Navascués_2008, Pironio_2010}.
	Given a state $\varrho$
    and a set of observables 
    $\{P_1, \dots, P_m\}$,
	one considers a moment matrix $\Gamma$ with entries
	\begin{equation}\label{eq:pseudo moment matrix}
		\Gamma_{\alpha\beta}=\tr ( P_\alpha^\dag P_\beta \varrho)\,.
	\end{equation}
	It can be shown that
	$\Gamma$ is positive semidefinite, $\Gamma \succeq 0$,
    while satisfying further relations arising from the algebraic relations 
    among the observables $P_\a$.
	
	The key idea of the moment hierarchy is
	to optimize over the set of \textit{pseudo moment matrices}. 
	These are positive semidefinite operators $\Delta \succeq 0$ that satisfy the same linear relations that the $\Gamma$ do,
	but which do not necessarily originate from
	a state.  Thus, optimizing over pseudo-moment matrices $\Delta$ is a relaxation
	of optimizing over the set of moment matrices $\Gamma$.
	Importantly, this optimization can be performed with semidefinite programming~\cite{boyd2004convex}. 
	Specifically, the ground state energy $E_0$ of a Hamiltonian can be lower bounded by
	\begin{align}\label{eq:pseudo_min}
		E_0 \quad \geq \min_{
			\substack{\Delta \text{ a pseudo-MM}}
		}
		\tr(C \Delta)
		\, ,
	\end{align}
	where $C$ is chosen such that $ \sum_{\a\b} C_{\a\b}P_\a^\dag P_\b = H$.
	We refer to Appendix~\ref{eq:lasserre} for further details.

	Surprisingly, one can also obtain {\it upper} bounds on the ground state energy from the solution
	of Eq.~\eqref{eq:pseudo_min}.
	This is done by
	rounding the optimal pseudo moment matrix $\Delta^*$
	to a physical state (see  Fig.~\ref{fig:illustration_rounding}).
	Combining upper and lower bounds on the optimum
	one can then obtain an approximation ratio.
	Various methods exist for doing this~ \cite{ParekhAlmost2021, anshu2020beyond, ParekhApplication2021, lee2022optimizing, king2022improved, lee2024improved, gribling2025improved, apte2025improvedalgorithmsquantummaxcut, bakshi2026sharp}.
	Here we follow the original
	approach of Ref.~\cite{ParekhAlmost2021}: 
	given a real pseudo moment matrix $\Delta^*$ whose entries are indexed by the set of single-qubit Pauli operators,
	$P = \{X_1, Y_1, Z_1, \dots, X_n, Y_n, Z_n\}$,
	a product state $\varrho^{\rm p} = \bigotimes_{i=1}^n \varrho_i$ is constructed. Here,
	\begin{equation}\label{eq:bloch_state}
		\varrho_i = \frac{1}{2} \left(I + a^X_i X_i + a^Y_i Y_i + a^Z_i Z_i\right)\, .
	\end{equation}

	To determine the coefficients $a^X_i, a^Y_i, a^Z_i$, compute the Cholesky decomposition of $\Delta^*$, which returns a real matrix $L$ such that $\Delta^* = L L^T$.
	Divide the entries of $L$ by $\sqrt{3}$.
	For each qubit $i$, 
    concatenate 
    the $3$ columns of $L$ corresponding to $X_i, Y_i, Z_i$, 
    to 
    construct a set of unit vectors $\mathbf{x}_i \in \mathds{R}^{9n}$ \cite{ParekhAlmost2021}.
 	Now sample a random matrix $R \in \mathds{R}^{3\times 9n}$
 	with mean zero and standard deviation one. Then set,
	\begin{equation}\label{eq:proj_moment_matrix}
		\mathbf{v}_i = \frac{R\mathbf{x}_i}{\|R\mathbf{x}_i\|} = \begin{pmatrix}
			a^X_i\\
			a^Y_i\\
			a^Z_i
		\end{pmatrix}\,.
	\end{equation}
	It is known that substituting these coefficients into
	Eq.\eqref{eq:bloch_state} returns at least $0.498$
	of the true ground state energy~\cite{ParekhAlmost2021}.

\smallskip
\noindent
{\bf Approximations via annealing. --- }
Analog quantum platforms offer heuristic tools to estimate the ground state energy of Hamiltonians.
In this paper we take the approach of encoding QMC into an experimentally more accessible \textit{proxy-Hamiltonian}.
Specifically, 
we consider the XY configuration
of a Rydberg atom array.
The Rydberg atoms form spin states 
which are coupled via dipole-dipole interactions \cite{henriet2020quantum, chen2023continuous, bornet2023scalable, titusfloquet}.
The Hamiltonian reads
\begin{equation}\label{eq:Rydberg_Hamiltonian}
	H^{\rm Ryd} = \underbrace{\sum_{i<j} \frac{C_3}{2|\mathbf{r}_i - \mathbf{r}_j|^3} \left( X_i X_j + Y_i Y_j \right)}_{H^{\rm r}} +\underbrace{\frac{\Omega(t)}{2}\sum_i X_i}_{H(t)}\, .
\end{equation}
Here the sum is over all atoms in the array, $C_3$ is a constant which is determined by the interaction strength
and $\Omega(t)$ is the Rabi frequency of an external field.  
Note that the static term  $H^{\rm r}$ from the Rydberg Hamiltonian shares two terms with $H^{\rm qmc}$ of Eq.~\eqref{eq:QMc_Hamiltonian}.
It is thus a natural idea to employ Rydberg atom arrays for approximations of quantum Max Cut.

Suppose that the positions $\mathbf{r}_i$ are chosen such that,
\begin{equation}
 \frac{C_3}{2|\mathbf{r}_i - \mathbf{r}_j|^3} = \omega_{ij}\,.
\end{equation}
This turns the static term in Eq.~\eqref{eq:Rydberg_Hamiltonian}
into
\begin{equation}\label{eq:static_rydberg}
 H^{\rm r} = \sum_{(ij) \in E} \omega_{ij}h^{\rm r}_{ij}\,, \quad\quad 
 h^{\rm r}_{ij} = X_i X_j + Y_i Y_j\,.
\end{equation}

Now let $\varrho^{\rm r}$ and $\varrho^{\rm qmc}$
be the ground states of $H^{\rm r}$ and $H^{\rm qmc}$
respectively so that $\tr(H^{\rm qmc} \varrho^{\rm qmc})$ is the solution of QMC.
While it is trivial to obtain the 
upper bound
$
    \tr(H^{\rm qmc} \varrho^{\rm qmc})\leq \tr(H^{\rm qmc} \varrho^{\rm r})$,
it is also possible to 
derive a {\it lower} bound on $\tr(H^{\rm qmc} \varrho^{\rm qmc})$.
The idea is to write the ground state energy of
$H^{\rm qmc}$ as a function of
$\tr(H^{\rm r} \varrho^{\rm qmc})$
and compare the result with $\tr(H^{\rm r} \varrho^{\rm r})$.

For this, note that the Hamiltonian $H^{\rm qmc}$ is invariant under the diagonal conjugate action of $SU(2)$.
In particular,
it is invariant under the subgroup generated by the Hadamard and phase gates $\langle H, S\rangle$,
which (up to a sign) permute the Pauli matrices.
Then, there always exists an invariant ground state
$\varrho^{\rm qmc}$ 
so that
\begin{equation}
	\tr(X_iX_j \varrho^{\rm qmc}) = \tr(Y_iY_j\varrho^{\rm qmc} ) = \tr(Z_iZ_j\varrho^{\rm qmc} )\,.
\end{equation}
This allows one to write
\begin{align}
		&\tr\big((X_iX_j + Y_iY_j + Z_iZ_j)\varrho^{\rm qmc} \big)
        \nn\\ 
        &=\frac{3}{2}\tr\big((X_iX_j + Y_iY_j) \varrho^{\rm qmc} \big) 
		= \frac{3}{2}\tr( h^{\rm r}_{ij} \varrho^{\rm qmc}) \, .
\end{align}
As a consequence,
\begin{align}
	\tr(H^{\rm qmc} \varrho^{\rm qmc})
	&= -\frac{1}{4}\sum_{(ij)\in E} \omega_{ij} \big( 1-\frac{3}{2}\tr( h^{\rm r}_{ij} \varrho^{\rm qmc}) \big)\,.
\end{align}
But since $\varrho^{\rm r}$ is the ground state of
$H^{\rm r}$, one has
$
 \tr(H^{\rm r}\varrho^{\rm r})
 \leq
 \tr(H^{\rm r} \varrho^{\rm qmc})\,,
$
from which follows that
\begin{equation}\label{eq:bound}
	\tr(H^{\rm qmc} \varrho^{\rm qmc}) \geq -\frac{1}{4}\sum_{(ij)\in E} \omega_{ij} \big( 1-\frac{3}{2}\tr( h^{\rm r}_{ij} \varrho^{\rm r}) \big)\,.
\end{equation}
The lower bound in Eq.~\eqref{eq:bound} will be a key ingredient for deriving the approximation ratio in our proposal.

Note that to map the QMC problem to a Rydberg atom array, 
one has to determine or positions $\mathbf{r_i}$ such that $\omega_{ij} = \frac{C_3}{2|\mathbf{r_i}-\mathbf{r_j}|^3}$. 
This is not always exactly possible,
but one can adjust the interaction terms further
by employing techniques such as
Floquet engineering, state masking, or coherent transport~\cite{
glaetzle2015designing, bluvstein2022quantum, 
Whitlock_2017, titusfloquet}.
In addition to this imperfect 
coefficient-to-distance mapping, 
the preparation of ground states via annealing is hampered by annealing time
and environmental noise. 
However, if the deviation from the true ground state energy
is not too large,
it is still possible to obtain a meaningful
approximation ratio (see Section {\it Robustness}).

\smallskip
\noindent
{\bf Hybrid algorithm. --- }
	We now introduce a hybrid algorithm
	that combines the
	previously described
	quantum and classical approaches.
	This works in the following way:
	first place $n$
	atoms in a Rydberg atom array
	at positions $\mathbf{r}_i$
	so to recreate
	the weights $\omega_{ij}$ of a given quantum Max Cut instance.
	Run an annealing procedure to find the ground state $\varrho^{\rm r}$ of the Rydberg Hamiltonian $H^{\rm r}$ from Eq.~\eqref{eq:static_rydberg}.
	Then 
    for all edges
	$(i j)\in E$
    measure 
	$\expect{X_iX_j+Y_iY_j}$. Here we denote by $\langle O \rangle = \tr(O \varrho^{\rm r})$ the expectation value of an operator $O$. 

	Second, construct a product state $\varrho^{\rm p}$ by rounding a pseudo moment matrix
	which has the observed values
	$\langle X_i X_j\rangle$ and $\langle Y_iY_j \rangle$ on the relevant edges
	as entries.
	To this end we consider pseudo-moment matrices whose entries are indexed by
	$P = \{X_1, Y_1, X_2, \dots , X_n, Y_n\}$.
	The measured expectation values can then be used as constraints in a semidefinite program, 
	\begin{align}\label{eq:SDP_Rydberg}
		\min_{\Delta} \quad & \tr(C \Delta) \nn\\
		\text{s.t.} \quad 
        & \Delta \succeq 0, \nonumber \nn\\
        & \Delta_{\a\a} = 1 && \text{for } \quad\a=\{1, \dots, 2n\}\,, \nn\\
        & \Delta_{\a\b} = \pm \Delta_{\a' \b'}
        && \text{if}\quad P_\a^\dagger P_\b = \pm P_{\a'}^\dagger P_{\b'}\,, \nonumber\\
		& \Delta_{\a\b} = \tfrac{1}{2}\langle X_i X_j+Y_iY_j\rangle
		&& \text{if}\quad P_\a^\dagger P_\b = X_iX_j \nn\\
		&&&
		\text{or}\quad P_\a^\dagger P_\b = Y_iY_j.
	\end{align}
	The pseudo moment matrix which solves this program is rounded into a product state $\varrho^{\rm p}$. The methodology is sketched in the section \textit{Product state energy}, with further details in Appendix~\ref{sec:proofs}.

	Finally, return whichever state $\varrho^{\rm r}$ or $\varrho^{\rm p}$ has lower energy. We summarize:
	\newcounter{algcount}
	\newcommand{\alglabel}[1]{\refstepcounter{algcount}\label{#1}Algorithm \thealgcount}
	\begin{description}
		\item[\alglabel{alg:rounding} (Hybrid QMC)]\hskip0pt\\
		\hrule

		{\bf 1.} Prepare the ground state $\varrho^{\rm r}$ of the Rydberg Hamiltonian $H^{\rm r}$ via quantum annealing.
        Measure $\expect{X_iX_j + Y_iY_j}$ for all edges $(ij) \in E$.

		\smallskip
		{\bf 2.} 
        Solve the semidefinite program in Eq.~\eqref{eq:SDP_Rydberg}
        with the constraints arising from
        $\varrho^{\rm r}$.

        \smallskip
		{\bf 3.}
        Round the solution $\Delta^* $
		into a product state $\varrho^{\rm p}$. Refer to Appendix~\ref{sec:proofs} for precise implementation.

		\smallskip
		{\bf 4.} Return $\arg\min\{\,\operatorname{tr}(H^{\rm qmc}\varrho^r),\operatorname{tr}(H^{\rm qmc}\varrho^p)\,\}$.

		\smallskip
		\hrule
	\end{description}

	Note that Algorithm~\ref{alg:rounding} requires a Rydberg atom system as quantum resource. However, the classical optimization is performed over $\mathcal{O}(n^2)$ real variables, whereas purely classical schemes that lead to entangled states require
	$\mathcal{O}(n^4)$ real variables.

\smallskip
\noindent
{\bf Approximation ratios. --- }
	An approximation algorithm is a polynomial-time algorithm
	that finds an approximation to a computationally hard problem~\cite{CLRS2022}.
	When trying to maximize a function
	over a set of instances (here on graphs $G$),
	it is common to judge its performance by an approximation ratio,
	\begin{equation}\label{eq:approx_ratio}
		\alpha := \min_{G} \frac{\text{ALG}(G)}{\text{OPT}(G)} \quad\in [0,1]\, .
	\end{equation}
	The function $\text{ALG}(G)$ is the approximation returned by the algorithm for an instance $G$ and $\text{OPT}(G)$ its optimal solution.
	This captures the worst case performance of an approximation algorithm, 
    so that $\text{ALG}(G) \geq \alpha~\text{OPT}(G) $ holds for every instance $G$. 

	In practice,
	to derive an approximation ratio for a function on a graph,
	one can introduce
	functions $\text{alg}(i, j)$ and $\text{opt}(i, j)$ for each edge $(ij)$ such that,
	\begin{equation}
	\text{ALG}(G) \geq \sum_{(i j)\in E}\text{alg}(i, j)\,, \quad
	\text{OPT}(G) \leq
	 \sum_{(i j)\in E}\text{opt}(i, j)\,.
	\end{equation}
	It is clear that for all~$G$,
	\begin{equation}\label{eq:lower_ratio}
	 \frac{\text{ALG}(G)}{\text{OPT}(G)}
	\geq
	 \frac{\sum_{(i j)\in E}\text{alg}(i, j)}{\sum_{(i j)\in E} \text{opt}(i, j)}
	 \geq 
	 \min_{(i j)\in E} \frac{\text{alg}(i, j)}{\text{opt}(i, j)}\,.
	\end{equation}

	We will use this strategy
	in Section {\it Approximation ratio of Algorithm~\ref{alg:rounding}} to find guarantees.
	In our context it is natural to define a {\it conditional approximation ratio}.
	
	\begin{definition}[Conditional approximation ratio]
		Consider an algorithm 
        $\textup{ALG}$
        that for each $G$ 
        has access to the optimal point
        $p^*(G)$ of a secondary hard problem.
		Then $\textup{ALG}$  has a
		\emph{conditional approximation ratio}
		$\kappa$ if
		\begin{equation}
			\min_{G} \,\, \frac{\textup{ALG}(G, p^*(G))}{\textup{OPT}(G)} \geq \kappa \, .
		\end{equation}
		\end{definition}

\smallskip
\noindent
{\bf Product state energy. --- }
    We first need to compute the energy of the product state $\varrho^{\rm p}$ from Algorithm~\ref{alg:rounding}.

    The optimization from Eq.~\eqref{eq:SDP_Rydberg} finds a pseudo moment matrix $\Delta^*$ that achieves the experimentally observed value $\tr(H^{\rm r} \varrho^{\rm r})$.
	Next, we compute the Cholesky decomposition of $\Delta^*\in \mathds{R}^{2n\times 2n}$, which returns a real matrix $L$ such that $\Delta^* = L L^T$. We construct the vectors  $\mathbf{x}_i\in \mathds{R}^{4n}$ obtained by concatenating the two columns of $L$ associated to the operators $X_i$ and $Y_i$.
    These are then normalized and rounded with a random matrix $R\in \mathds{R}^{3 \times 4n}$
	as in Eq.~\eqref{eq:proj_moment_matrix}.
			
		\begin{restatable}[]{observation}{prodenergy}
			\label{obs:prod_energy}
			Let $\varrho^{\rm p}$ be the product state obtained from rounding the solution $\Delta^*$ of
			Eq.~\eqref{eq:SDP_Rydberg}. Then
	\begin{align}\label{eq:prod_energy}
		\mathbb{E}\left[\tr(H^{\rm qmc} \varrho^{\rm p})\right] &=
		-\frac{1}{4}\sum_{(ij)\in E}\omega_{ij}\Big(1-\frac{8}{3\pi} \hat{F}(t_{ij})\Big)\,,\\
		\hat{F}(t_{ij}) &= t_{ij}\cdot {}_2F_1 \bigg(
		\begin{matrix} 1/2, 1/2 \\ 5/2 \end{matrix} 
		; t^2_{ij} 
		\bigg)\,,\nn
	\end{align}
	where
	$t_{ij}= \frac{1}{2}\expect{X_iX_j + Y_iY_j}$ and
	${}_2F_1$ is the hypergeometric function.
	\end{restatable}
	The proof can be found in Appendix~\ref{sec:proofs}.
	Moreover, Appendix~\ref{sec:2_3} proves the {\it existence} of a product state attaining at least
	$-\frac{1}{4}\sum_{(ij)\in E}\omega_{ij}(1-\frac{1}{2}t_{ij})$, for which we however do not know an algorithm.

\smallskip
\noindent
{\bf Upper and lower bounds. --- }
	To obtain a conditional approximation ratio for Algorithm~\ref{alg:rounding}, we need
upper and lower bounds on the ground state energy.
When considering the {\it minimization} of a negative function, we can apply the same analysis from Eq.~\eqref{eq:lower_ratio}.
The difference is that one needs an
{\it upper bound} on the approximation
and a {\it lower bound} on the optimum.

    The lower bound on QMC is given by Eq.~\eqref{eq:bound} using the measurements  $\{\expect{X_iX_j + Y_iY_j}\}_{(ij)\in E}$.
    The
	upper bound 
	is obtained from the energies
	$\tr(H^{\rm qmc} \varrho^{\rm p})$ and
	$\tr(H^{\rm qmc} \varrho^{\rm r})$.
	We computed the energy $\tr(H^{\rm qmc} \varrho^{\rm p})$ in the previous section, now we focus on
	$\tr(H^{\rm qmc} \varrho^{\rm r})$.
	
To upper bound the energy
of $\varrho^{\rm r}$
	it suffices 
	to upper bound 
	$\expect{Z_iZ_j}$
	given the value of
	$\expect{X_iX_j + Y_iY_j}$.

	\begin{restatable}[]{lemma}{boundZiZj}\label{theo:bound_ZiZj}
		For all states it holds that
		\begin{align}
			 \expect{Z_iZ_j}\leq 1 - 2|t_{ij}|, \quad t_{ij} = \frac{1}{2}\expect{X_iX_j + Y_iY_j}\,.
		\end{align}
	\end{restatable}
 The proof is in Appendix \ref{proof:zizj}.
	Now recall that the QMC energy of
	 $\varrho^{\rm r}$ is
	\begin{equation}
		\tr(H^{\rm qmc} \varrho^{\rm r}) =
		-\frac{1}{4}\sum_{(ij)\in E} \omega_{ij}\expect{I - X_iX_j- Y_iY_j - Z_iZ_j}\, .
	\end{equation}
	Applying Lemma~\ref{theo:bound_ZiZj} yields the upper bound
	\begin{align}\label{eq:energy_Rydberg}
		\tr(H^{\rm qmc}\varrho^{\rm r} ) &\leq  -\frac{1}{4}\sum_{(ij) \in E} \omega_{ij} \Big(1-2t_{ij} - \underbrace{(1-2|t_{ij}|)}_{ \geq \tr(Z_iZ_j \varrho^{\rm r}) }\Big)\nn\\
		&= \sum_{(ij) \in E} \omega_{ij} \min\big\{t_{ij}, 0\big\} \,,
	\end{align}
	where $t_{ij} = \frac{1}{2}\expect{X_iX_j + Y_iY_j}$, with $t_{ij}\in [-1, 1]$. 
	
\smallskip
\noindent {\bf Approximation ratio of Algorithm~\ref{alg:rounding}. --- }
    Algorithm~\ref{alg:rounding}  requires access to the ground state $\varrho^{\rm r}$
       of the Hamiltonian $H^{\rm r}$. 
        This lower bounds QMC by Eq.~\eqref{eq:bound}. 
        The approximation ratio of Algorithm~\ref{alg:rounding} is then obtained by comparing this bound with
        the approximation from Observation~\ref{obs:prod_energy} and the upper bound of Eq.~\eqref{eq:energy_Rydberg}.
	
	\begin{theorem}\label{eq:main_theorem}
		Algorithm~\ref{alg:rounding} has a conditional approximation ratio of $\kappa = 0.651$.
	\end{theorem}
	\begin{proof}
	The approximation ratio requires suitable
	lower bounds on the optimum and
	upper bounds on the approximation.
	First, recall the lower bound from Eq.~\eqref{eq:bound} on the optimum.
	Second, obtain an upper bound on the approximation
	${\rm arg}\min\{\tr(H^{\rm qmc}\varrho^{\rm r}),
	\tr(H^{\rm qmc}) \varrho^{\rm p}\}$
	by the convex combination,
	\begin{align}\label{eq:approx1}
		&{\rm arg}\min\big\{\tr(H^{\rm qmc} \varrho^{\rm r}) , \tr(H^{\rm qmc}\varrho^{\rm p} )\big \} \nn\\ &\quad\quad   \leq p \tr(H^{\rm qmc} \varrho^{\rm r} )
			+ (1-p) \tr(H^{\rm qmc}\varrho^{\rm p})\,,
	\end{align}
	where $p\in [0, 1]$.
        
        Therefore, by Eqs.~\eqref{eq:bound},~\eqref{eq:lower_ratio},~\eqref{eq:approx1},		
		\begin{align}\label{eq:convex_energy}
			\kappa &\geq \frac{{\rm arg}\min\{\tr(H^{\rm qmc} \varrho^{\rm r} ),\tr(H^{\rm qmc}\varrho^{\rm p}\})}{\tr(H^{\rm qmc}\varrho^{\rm qmc})}\nn\\
			&\geq \max_{p\in [0, 1]}\frac{p \tr(H^{\rm qmc} \varrho^{\rm r} ) + (1-p) \tr(H^{\rm qmc}\varrho^{\rm p})}{-\frac{1}{4}\sum_{(ij)\in E}\omega_{ij} (1 - 3 t_{ij})}\,,
		\end{align}
		where $t_{ij} = 
        \frac{1}{2} \expect{X_iX_j + Y_iY_j}= 
        \tfrac{1}{2}
        \tr(h_{ij} \varrho^{\rm r})$.
		By Eq.~\eqref{eq:bound} we see that for $t_{ij}\geq \frac{1}{3}$ the edge $(ij)$ has a positive value.

		However, to compute the approximation ratio using Eq.~\eqref{eq:convex_energy},
		the values on all edges would need to be non-positive. We consider the weaker lower bound
		\begin{align}\label{eq:weaker_bound}
			\tr(H^{\rm qmc}\varrho^{\rm qmc}) &\geq -\frac{1}{4}\!\!\sum_{(ij) \in E} 1 - 3t_{ij} \nn\\
			&\geq -\frac{1}{4}\!\!\sum_{(ij)\in E} {\rm arg}\max\{1- 3t_{ij}, 0\}\, .
		\end{align}
		Then by Eq.~\eqref{eq:lower_ratio},
		the ratio is at least that of the worst performing edge
        for which we use the variable 
        $t\in [-1, \frac{1}{3}]$ 

        Inserting in Eq.~\eqref{eq:convex_energy} the 
        energies
        $\varrho^{\rm r}$
        [Eq.~\eqref{eq:energy_Rydberg}] 
        and
        of $\varrho^{\rm p}$
        [Eq.~\eqref{eq:prod_energy}]
        yields numerically
		\begin{align}\label{eq:ratio_alg}
			\kappa &\geq \max_{p\in [0, 1]}\min_{t\in [-1, \frac{1}{3})} 
            \bigg\{
            p\frac{-4t}{1-3t} 
			+ (1-p) 
			\frac{1-\frac{8}{3\pi} \hat{F}(t)}{1- 3 t} \bigg\}
            \nn\\ &\approx 0.651 \, .
		\end{align}
		This ends the proof.
\end{proof}
	 We refer to Appendix~\ref{append:numerics} for numerical results regarding the performance of Algorithm~\ref{alg:rounding} in specific graphs.

\smallskip
\noindent{\bf Robustness. ---}
Theorem~\ref{eq:main_theorem} requires access to the exact ground state of $H^{\rm r}$ [Eq.~\eqref{eq:static_rydberg}]. 
To assess the performance of Algorithm~\ref{alg:rounding} in noisy hardware, we compute the regime for which the method
yields better guarantees than purely classical solutions.
    
Let $\varrho^A$ be a state such that $\tr(H^{\rm r}\varrho^A ) = \eta \tr(H^{\rm r}\varrho^{\rm r})$
for $\eta\in [0, 1]$.
Using Eq.~\eqref{eq:bound} we have,
		\begin{equation}\label{eq:noisy_bound}
			\tr(H^{\rm qmc} \varrho^{\rm qmc}) \geq -\frac{1}{4}\sum_{(ij)\in E} \omega_{ij} \left(1-\frac{3}{\eta}\, t_{ij}\right)\,,
		\end{equation}
		 where $t_{ij} = \frac{1}{2}\tr((X_iX_j + Y_iY_j)\varrho^A )$. 
		Then, following the proof from Theorem~\ref{eq:main_theorem}, we combine the lower bound Eq.~\eqref{eq:noisy_bound} and the upper bound Eqs.~\eqref{eq:prod_energy},~\eqref{eq:energy_Rydberg}. One sees that $ \kappa = 0.614$ for $\eta = 0.89$.

\smallskip
\noindent
	{\bf Outlook. --- } 
	Our hybrid approximation algorithm for quantum Max Cut surpasses
	the currently best-known classical algorithm,
	even when the quantum computer
	achieves only $0.89$ of the Rydberg ground state energy.
	Moreover, while classical algorithms are focusing on solving an SDP with $\mathcal{O}(n^4)$ variables,
	our approach requires $\mathcal{O}(n^2)$
	when combined with quantum resources.

	Our work raises several natural questions:
	What ratios can be obtained in practice
	with noisy hardware for studying large systems?
	Can other platforms be used for similar approximations?
	
	\bibliographystyle{apsrev4-2}
	\bibliography{bibliography}
	\newpage
	\appendix
	
	\setcounter{secnumdepth}{1} 
	\makeatletter
	\renewcommand{\@hangfrom@section}[2]{#1\ \ #2: }
	\makeatother

		\section{Numerical results}\label{append:numerics}
		\subsection{Approximation for random graphs}
		The main text introduced Algorithm~\ref{alg:rounding}
		to approximate the ground state energy of the QMC Hamiltonian,
		\begin{equation}
			H^{\rm qmc} = - \frac{1}{4}\sum_{(ij)\in E} \omega_{ij} (I - X_iX_j - Y_iY_j - Z_iZ_j)\, .
		\end{equation}
		We compared the upper bound from Eqs.~\eqref{eq:prod_energy},~\eqref{eq:energy_Rydberg} with the lower bound from Eq.~\eqref{eq:weaker_bound}.
		The approximation ratio of Algorithm~\ref{alg:rounding} is then obtained by minimizing the function in Eq.~\eqref{eq:ratio_alg}.  Fig.~\ref{fig:ratios_per_edge} shows that the convex combination $\varrho = p\varrho^{\rm p} + (1-p) \varrho^{\rm r}$ yields the lower bound $\kappa=0.651$ at $p=0.345$.
		
		\begin{figure}[htb]
			\centering
			\includegraphics[width=\linewidth]{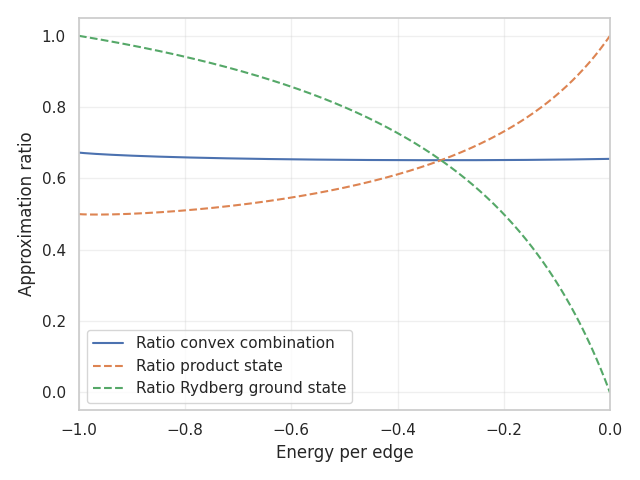}
			\caption{
			Lower bound on the approximation ratio versus the energy of the worst-performing edge. The states plotted are $\varrho^{\rm r}$, $\varrho^{\rm p}$, and $\varrho = p \varrho^{\rm r} + (1-p) \varrho^{\rm p}$ at $p=0.345$. For the mixed state, the curve achieves an approximation ratio of $\kappa = 0.651$.
			}
			\label{fig:ratios_per_edge}
		\end{figure}

	However, the performance of
	Algorithm~\ref{alg:rounding} often surpasses the specific ratio given by Theorem~\ref{eq:main_theorem}.
	To see this, we compute
	for small graphs
	numerically the ratio
	\begin{equation}\label{eq:ratio_append}
		\mu = \frac{\tr(H^{\rm qmc} \varrho^{\rm r})}{\tr(H^{\rm qmc}\varrho^{\rm qmc})}\,.
	\end{equation}
	Here $\varrho^{\rm r}$ and $\varrho^{\rm qmc}$ denote the ground states of the Hamiltonians $H^{\rm r}$ and $H^{\rm qmc}$, respectively.
	Specifically, consider a fully connected graph with independent and uniformly random weights in the interval
	$\omega_{ij}\in [0, 1]$.
	We will only consider the coefficients for which
	$H^{\rm r}$ presents non-degenerate ground states.
	The reason is that computing $\mu$ requires evaluating $\tr(Z_iZ_j \varrho^{\rm r})$. In the degenerate case, this quantity can depend on the ground state selected.
	Therefore, while Rydberg Hamiltonians with degenerate ground states may be of experimental interest, it is difficult to determine the worst-case performance in such instances using exact diagonalization.~\footnote{In order to find which ground state is returned by the quantum platform, one could simulate the time evolution of the Rydberg array \cite{Silverio2022pulseropensource}.}

	Fig.~\ref{fig:random_ratios} shows the ratio $\mu$ for 1000 coefficients in $6$-, $12$- and $18$-qubit systems.
	The ratio is always above $0.95$,
	but the mean value of $\mu$ decreases with the number of qubits. This indicates that $\mu$ can exceed the approximation ratio by a large margin.  
	
	\begin{figure}[htb]
		\centering
		\includegraphics[width=\linewidth]{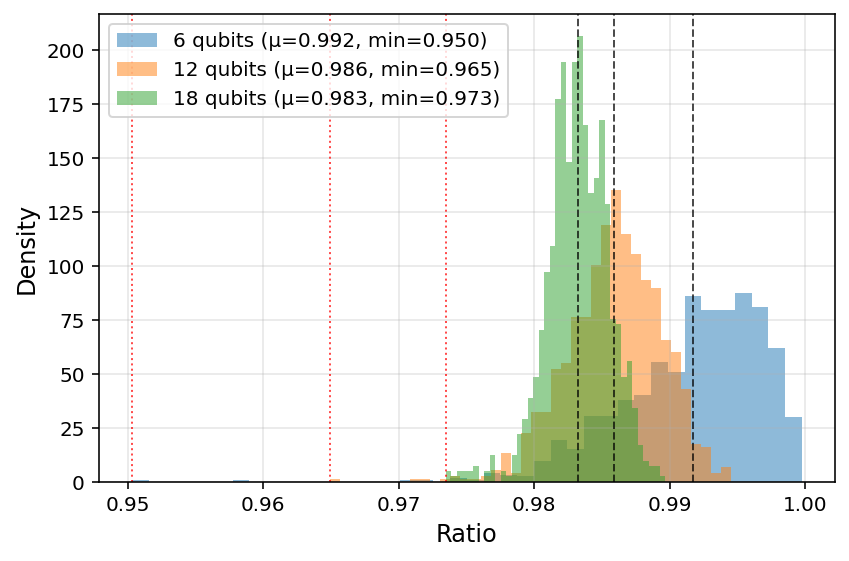}
		\caption{Ratio $\mu$ associated with $\varrho_{\rm r}$ for 1000 random graphs on $6$-, $12$- and $18$-qubit systems.}
		\label{fig:random_ratios}
	\end{figure}

	\subsection{Noise Robustness}
	In practice, current quantum platforms can obtain only approximations to the true ground states. This is due to limited annealing time and experimental imperfections, including decoherence and calibration errors.	

	In order to analyse the robustness of the algorithm with respect to the experimental noise we mimic experimental conditions by adding thermal noise. In that case we consider the resulting Gibbs states at a given temperature as the outcome of the annealing process. 
Specifically, we consider the Hamiltonian $H^{\rm r}$ with coefficients $\omega_{ij}\in [0, 1]$ drawn independently and uniformly random on a fully connected graph, the corresponding Gibbs state at temperature $T$ being given by,
	\begin{equation}
		\varrho^G = \frac{e^{-H^{\rm r}/T}}{Z}\,, \quad  \text{where} \quad Z = \tr\Big(e^{-H^{\rm r}/T}\Big)\, .
\end{equation}  
	For this state,
\begin{equation}
	\tr(\varrho^G H^{\rm r}) = \frac{\sum_{i} E_i e^{-E_i/T}}{\sum_{i} e^{-E_i/T}}\, ,
\end{equation}
where $E_i$ are the eigenvalues of $H^{\rm r}$.
\begin{figure}[tbp]
	\centering
	\includegraphics[width=\linewidth]{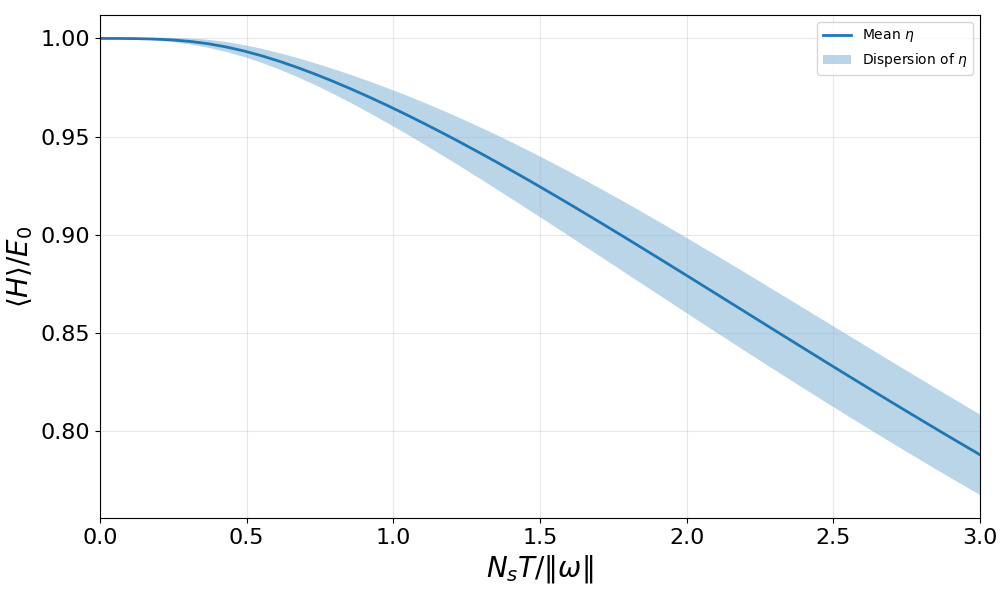}
	\caption{Ground state approximation of 1000 Gibbs states with temperatures $T\in \big(0, 3\|\omega\|/N_s\big]$ in a 12-qubit system, where $\|\omega\|$ is the Frobenius norm.}
	\label{fig:noise_numerics}
\end{figure}
In that case the deviation of the Rydberg hamiltonian ground-state energy  used in  Algorithm~\ref{alg:rounding} due to thermal noise is given by
\begin{equation}
	\eta = \frac{\tr(\varrho^G H^{\rm r})}{E_0}\,.
\end{equation}
Specifically, for 1000 sets of randomly chosen coefficients $\omega$ in a 12-qubit system, we compute $\eta$ for $T \in (0, 3\|\omega\|/N_s]$.
The temperature range investigated is considered significantly lower than the magnetic order-to-paramagnetic phase transition temperature which is of the order of the mean coupling/coefficient strength $\langle \omega_{ij} \rangle$ or equivalently, of the per-spin  Frobenius norm of the coupling matrix $||\omega||/N_s$, where $\|\omega\| = \sqrt{\sum_{ij} \omega_{ij}^2}$ and $N_s$ stands for the number of spins/qubits.

	The mean, minimum, and maximum ratios at each temperature are shown in Fig.~\ref{fig:noise_numerics}. We see that at low temperatures, $\eta$ stays close to $1$ before decreasing linearly with temperature. Notably, for all instances, if $T < 1.6\|\omega\|/N_s$ then $\eta \geq 0.89$, as required to surpass the best-known classical algorithm.

	\section{Moment-SOS hierarchy}

	\label{eq:lasserre}
	The moment-SOS hierarchy (also known as Navascues-Pironio-Acín or quantum Lasserre hierarchy) was introduced for noncommutative polynomial
	optimization in quantum mechanics
	\cite{Navascués_2008, Pironio_2010}, see also Helton~\cite{helton}.
	In particular, this hierarchy provides a sequence of outer approximations to the set of quantum states.
	Here we will employ the hierarchy through \textit{pseudo-moment matrices} to lower bound the ground state energy of a Hamiltonian $H$. 
	
	Consider the set of $m$ operators
	\begin{equation}\label{eq:subset_P}
		P = \{P_1, P_2, \dots, P_m\}\, .
	\end{equation}
	Given a state $\varrho$, let $\Gamma$ be a \textit{moment matrix} of dimension $m\times m$ with entries given by
	\begin{equation}
		\Gamma_{\alpha\beta} = \tr(P_\alpha^\dag P_\beta\varrho)\, .
	\end{equation}
	For simplicity, we can choose the operators $P_\alpha$ to be Pauli strings.
	By construction, the matrix $\Gamma$ is hermitian and satisfies
	\begin{align}\label{eq:cond_moment_matrix}
		&\Gamma_{\a\b} = \Gamma_{\a'\b'} \quad \quad &&\text{if} \quad \quad P_\a^\dag P_\b = P_{\a'}^\dag P_{\b'}\, ,\nn \\
		&\Gamma_{\a\b} = 1 &&\text{if} \quad \quad P_\a^\dag P_\b = I\, ,\nn \\
		&\Gamma \succeq 0 \, .
	\end{align}
	
	The moment-SOS hierarchy considers the set of all matrices $\Delta$ that satisfy the same linear relations as $\Gamma$. This provides a relaxation to the set of moment matrices, allowing one to lower bound the ground state energy of any Hamiltonian.
	Specifically, let $C$ be a matrix such that $\sum_{\a\b} C_{\a\b} P^\dag_\a P_\b = H$. Solving the optimization problem
	\begin{align}\label{eq:SDP}
		\min_{\Delta} \quad & \tr(C \, \Delta) \nn\\
		\text{s.t.} \quad 
		& \Delta \succeq 0, \nonumber\\
		&\Delta = \Delta^T\nn\\
		& \Delta_{\a\b}= \pm \Delta_{\a'\b'}
		&& \text{if } P_\a^\dagger P_\b = \pm P_{\a'}^\dagger P_{\b'}, \nonumber\\
		& \Delta_{\a\a} = 1 && \text{if } \a \in \{1, \dots, m\}\, ,
	\end{align}
	yields a lower bound on the ground state energy of $H$. 
	The tightness of this lower bound depends on the choice of the set $P$ from Eq.~\eqref{eq:subset_P}. 

	In particular, level $k$ of the moment-SOS hierarchy restricts the set $P$ to all Pauli strings acting non-trivially on at most $k$ qubits.
	The resulting pseudo moment matrix $\Delta$ has exponential size in the level of the hierarchy but is polynomial in the number of qubits.

	For the optimization, even low orders of the hierarchy come with an adverse computational scaling. Therefore, to study large-scale systems, one may consider only a subset of operators at each level of the hierarchy while exploiting the symmetries of the Hamiltonian \cite{wang2026scalable}.
	Nevertheless, the study of approximation ratios for quantum Max Cut has so far relied on the moment-SOS hierarchy, with symmetry exploitation playing a limited role.
	
	In particular, existing QMC approximation algorithms are typically based on the full second level of the Lasserre hierarchy. \cite{anshu2020beyond, ParekhApplication2021, lee2022optimizing, king2022improved, lee2024improved, gribling2025improved, apte2025improvedalgorithmsquantummaxcut, bakshi2026sharp}. In contrast,  Algorithm~\ref{alg:rounding} employs a first level of the Lasserre hierarchy with additional linear constraints. This difference results in a classically more efficient algorithm.

	\section{Proofs from main body}\label{sec:proofs}
	\subsection{Product state rounding}

	Algorithm~\ref{alg:rounding} achieves a $0.651$-approximation by returning
	whichever of the states $\varrho^{\rm r}$ and $\varrho^{\rm p}$ has lower energy.
	Here we will provide a precise description of how $\varrho^{\rm p}$ is constructed and prove Observation~\ref{obs:prod_energy}.

	Recall that the performance of Algorithm~\ref{alg:rounding} is judged by the conditional approximation ratio. Given access to the optimal point
	$p^*(G)$ of a secondary hard problem,
	\begin{equation}
		\min_{G} \,\, \frac{\textup{ALG}(G, p^*(G))}{\textup{OPT}(G)} \geq \kappa \, .
	\end{equation}
	Following Gharibian and Parekh \cite{ParekhAlmost2021}, our aim is to construct a product state
	$\varrho^{\mathrm{p}}$ from a lower bound on the optimal solution $\text{OPT}(G)$.
	This product state will have the form
	\begin{align}\label{eq:prod2}
		&\varrho^{\rm p} = \bigotimes_{i=1}^n \varrho_i\,,\nn\\
		&\varrho_i = \frac{1}{2} \left(I + a^X_i X_i + a^Y_i Y_i + a^Z_i Z_i\right)\, ,
	\end{align}
	where each $\mathbf{v}_i = (a^X_i, a^Y_i, a^Z_i)$ is a unit vector
	that is yet to be determined. Then,
	\begin{equation}\label{eq:energy_prod_vec}
		\tr(H^{\rm qmc} \varrho^{\rm p}) = -\frac{1}{4}\sum_{(i j)\in E}\omega_{ij} (1 - \mathbf{v}_i \cdot \mathbf{v}_j) \, .
	\end{equation}
	We construct the vectors $\mathbf{v}_i$ using a pseudo moment matrix $\Delta^*$ that encodes the measurements of a Rydberg array.
	To this end, index a pseudo moment matrix
	$\Delta \in \mathds{R}^{2n\times 2n}$
	by the set of operators
	\begin{equation}\label{eq:op_prod}
		P = \{X_1, Y_1, X_2, Y_2, \dots , X_n, Y_n\}\, .
	\end{equation}
	We demand that the observed Rydberg measurements
	$\expect{X_iX_j + Y_iY_j}$ appear in the relevant entries and optimize the functional representing the energy,
	\begin{align}\label{eq:optimization_SDP}
		\min_{\Delta} \quad & \tr(C \Delta) \nn\\
		\text{s.t.} \quad 
		& \Delta \succeq 0, \nonumber \nn\\
		& \Delta_{\a\a} = 1 && \text{for } \quad\a=\{1, \dots, 2n\}\,, \nn\\
		& \Delta_{\a\b} = \pm \Delta_{\a' \b'}
		&& \text{if}\quad P_\a^\dagger P_\b = \pm P_{\a'}^\dagger P_{\b'}\,, \nonumber\\
		& \Delta_{\a\b} = \tfrac{1}{2}\langle X_i X_j+Y_iY_j\rangle
		&& \text{if}\quad P_\a^\dagger P_\b = X_iX_j \nn\\
		&&&
		\text{or}\quad P_\a^\dagger P_\b = Y_iY_j,
	\end{align}
	where $C$ is such that $H^{\rm r}=\sum_{\a\b}C_{ij} P^\dag_i P_j$.  
	Necessarily, any solution $\Delta^*$
	to the SDP in Eq.~\eqref{eq:optimization_SDP} returns the observed ground state energy of the Rydberg Hamiltonian.
	
	We construct a set of unit vectors $\mathbf{v}_i\in \mathds{R}^3$ defining a product state by Eq.~\eqref{eq:prod2}.
	Perform a Cholesky decomposition of $\Delta^*$ to obtain a real matrix
	$L$ such that $\Delta^* = L L^T$. Note that the columns $L_\a$ are associated with the operators of $P$ from Eq.~\eqref{eq:op_prod} and satisfy
	\begin{equation}
		L_\a \cdot L_\b = \Delta^*_{\a\b}\,.
	\end{equation}
	The columns are unit vectors because $L_\a \cdot L_\a = \Delta_{\a\a} = 1$.
	
	Then, for each qubit $i$, concatenate the columns $L_\a$ associated with the operators $X_i$ and $Y_i$. Let $\mathbf{x}_i\in \mathds{R}^{4n}$ be the unit vector obtained by dividing the resulting vector by $\sqrt{2}$. By construction, these satisfy:
	\begin{equation}\label{eq:objective_prod}
		\mathbf{x}_i \mathbf{x}_j = \frac{1}{2}\tr((X_iX_j + Y_iY_j)\varrho^{\rm r})\,.
	\end{equation}
	Now note that these are unit vectors $\mathbf{x}_i$ have with dimension $4n$ instead of $3$.
	To construct Bloch vectors $\mathbf{v}_i \in \R^3$,
	we use the following result.
	\begin{lemma}[Bri\"et-de Oliveira Filho-Vallentin~\cite{Briet_2014}, Lemma 2.1]
	\label{lemma:briet}
		Let $\mathbf{x}_i, \mathbf{x}_j$ be unit vectors in $\mathds{R}^n$ and let $R\in \mathds{R}^{r\times n}$ be a random matrix whose entries are distributed independently according to the standard normal distribution with mean $0$ and variance $1$. Then,
		\begin{equation*}
			\mathds{E}\left[\frac{R\mathbf{x}_i}{\|R\mathbf{x}_i\|} \cdot \frac{R\mathbf{x}_j}{\|R\mathbf{x}_j\|} \right] = \frac{2}{r}\left(\frac{\Gamma((r+1)/2)}{\Gamma(r/2)}\right)^2 \hat{F}(\mathbf{x}_i\cdot \mathbf{x}_j)\, .
		\end{equation*}
		Here $\hat{F}(\mathbf{x}_i\cdot \mathbf{x}_j)$ is given by
		\begin{align*}
			&\hat{F}(\mathbf{x}_i\cdot \mathbf{x}_j)
			 = (\mathbf{x}_i\cdot \mathbf{x}_j)
			\,  {}_2F_1
			\bigg(
			\begin{matrix} 1/2, 1/2 \\ r/2+1 \end{matrix} 
			; (\mathbf{x}_i\cdot \mathbf{x}_j)^2
			\bigg)
			\nn\\
			&= \sum_{k=0}^\infty
			\frac{(1\cdot 3 \cdots (2k-1))^2}
			{(2\cdot 4 \cdots 2k)((r+2)\cdot(r+4)\cdots (r+2k))}
			(\mathbf{x}_i\cdot \mathbf{x}_j)^{2k+1}
		\end{align*}
		with ${}_2F_1$ the hypergeometric function.
	\end{lemma}
	This Lemma allows to evaluate the expected inner product of two vectors in $\mathds{R}^n$,
	after they have been randomly projected onto $\mathds{R}^r$ and subsequently normalized.
	
	We use this random projection to map each unit vector
	$\mathbf{x}_i \in \mathds{R}^{4n}$ into the coefficients $\mathbf{v}_i\in \mathds{R}^3$ of a one-qubit state,
	\begin{equation}\label{eq:proj_prods}
		\mathbf{v}_i =\frac{R\mathbf{x}_i}{\|R\mathbf{x}_i\|} = \begin{pmatrix}
			a^X_i\\
			a^Y_i\\
			a^Z_i
		\end{pmatrix}\,.
	\end{equation}
	Here $R\in \mathds{R}^{3\times 4n}$ is a matrix with entries sampled independently from a Gaussian distribution with mean 0 and standard deviation 1.
			
	\prodenergy*
	
	\begin{proof}
	This proof is equivalent to that from Ref.~\cite{ParekhAlmost2021}. However, we include it here for completeness.
	
	The product state has energy given by 
	\begin{equation}
		\tr(H^{\rm qmc} \varrho^{\rm p}) = -\frac{1}{4} \sum_{(ij)\in E} \omega_{ij} (1- \mathbf{v}_i\cdot \mathbf{v}_j)\, .
	\end{equation}
	By definition, we can express the vectors $\mathbf{v}_i$ using the randomized projection of $\mathbf{x}_i$:
	\begin{equation}
		\tr(H^{\rm qmc} \varrho^{\rm p}) = -\frac{1}{4} \sum_{(ij)\in E} \omega_{ij} \Bigg(1- \frac{R\mathbf{x}_i}{\|R\mathbf{x}_i\|}\cdot \frac{R\mathbf{x}_j}{\|R\mathbf{x}_j\|}\Bigg)\, .
	\end{equation}
	
	Recall that $R \in \mathds{R}^{3\times 4n}$. Thus, by Lemma~\ref{lemma:briet},
	\begin{equation}
		\mathbb{E}\left[\frac{R\mathbf{x}_i}{\|R\mathbf{x}_i\|} \cdot \frac{R\mathbf{x}_j}{\|R\mathbf{x}_j\|}\right] = \frac{8}{3\pi} \hat{F}(\mathbf{x}_i \cdot \mathbf{x}_j)\, .
	\end{equation}
	Then, the expected energy for the state $\varrho^{\rm p}$ is
	\begin{equation}
		\mathbb{E}\left[\tr(H^{\rm qmc} \varrho^{\rm p})\right] = -\frac{1}{4} \sum_{(ij)\in E}\omega_{ij}\left(1-\frac{8}{3\pi} \hat{F}(t_{ij})\right)\, ,
	\end{equation}
	where $t_{ij}= \mathbf{x}_i\cdot \mathbf{x}_j = \frac{1}{2}\expect{X_iX_j + Y_iY_j}$.
	\end{proof}

	\subsection{Two-body marginals}\label{proof:zizj}
	
	\boundZiZj*
	
	\begin{proof}
		Consider a state $\varrho$. For a given pair of qubits $i,j$, let $\Delta$ be the pseudo moment matrix indexed by the set
		\begin{equation}
			P = \bigg\{I, X_iX_j, Y_iY_j, Z_iZ_j\bigg\}\, ,
		\end{equation}
		with entries given by $\Delta_{\a \b} = \tr( P_\a^{\dag}P_\b\varrho)$. Thus,
		\begin{equation}
			\Delta = \begin{pmatrix}
				1 & \expect{X_iX_j} & \expect{Y_iY_j} &\expect{Z_iZ_j}\\
				\expect{X_iX_j} & 1 & -\expect{Z_iZ_j} & -\expect{Y_iY_j}\\
				\expect{Y_iY_j} & -\expect{Z_iZ_j} & 1 & -\expect{X_iX_j}\\
				\expect{Z_iZ_j} & -\expect{Y_iY_j} & -\expect{X_iX_j} & 1
			\end{pmatrix}\, .
		\end{equation}
		We now want to upper bound the expectation value $\expect{Z_iZ_j}$ using  $t_{ij}= \frac{1}{2}\expect{X_iX_j+Y_iY_j}$. In other words, given $\expect{X_iX_j + Y_iY_j} = 2t_{ij}$, maximize $\expect{Z_iZ_j}$.
	
		Consider the following moment matrix relaxation.
		\begin{align}
			\max \quad & \Delta_{03} \label{eq:opt_ZiZj} \\[4mm] 
			\text{such that} \quad & \Delta \succeq 0 \tag{\theequation.1} \label{append:cond_1}\\
			& \Delta = \Delta^T \tag{\theequation.2}\label{append:cond_2} \\
			& \Delta_{01} = -\Delta_{32} \tag{\theequation.3} \label{append:cond_3}\\
			& \Delta_{02} = -\Delta_{31} \tag{\theequation.4} \label{append:cond_4}\\
			& \Delta_{03} = -\Delta_{12} \tag{\theequation.5} \label{append:cond_5}\\
			& \Delta_{01} + \Delta_{02} = 2t_{ij} \tag{\theequation.6} \label{append:cond_6}
		\end{align}
		In order to obtain analytical bounds for $\expect{Z_iZ_j}$ using the moment matrix formalism, we will symmetry-reduce Eq.~\eqref{eq:opt_ZiZj}.
		Note that the constraint set and the objective function of the SDP~\eqref{eq:opt_ZiZj}
		are invariant under the unitary operation
		\begin{equation}\label{eq:inv_map}
		X \mapsto Y , \quad\quad Y\mapsto -X , \quad\quad Z\mapsto Z	\, .
		\end{equation}
		We can thus consider the optimization
		\begin{align}\label{eq:pseudo-zizj}
			\max \quad & \c\nn\\[2mm]
			\text{such that} \quad & -1\leq t_{ij} \leq 1\\
			&\begin{pmatrix}
				1 & t_{ij} & t_{ij} & \c\\
				t_{ij} & 1 & -\c & -t_{ij}\\
				t_{ij} & -\c & 1 & -t_{ij}\\
				\c & -t_{ij} & -t_{ij} & 1
			\end{pmatrix}\succeq 0 \, .
		\end{align}
		which by standard arguments (invariant feasibile set and invariant objective function, see e.g. \cite{Gatermann_2004}),
		achieves the same objective value as the SDP in Eq.~\eqref{eq:opt_ZiZj}. Here the entry $\c$ corresponds to $\expect{Z_iZ_j}$.

		Note that we can now compute the eigenvalues of $\Delta$ exactly,
		\begin{align}\label{eq:eigen_append}
			\lambda_1 &= 1+2t_{ij}-\c \,,\nn\\
			\lambda_2 &=  1-2t_{ij}-\c \,,\nn\\
			\lambda_3 &= \lambda_4 = 1+\c \, .
		\end{align}
		Recall that a bound on $\expect{Z_iZ_j}$ is obtained by maximizing $\c$.
	The eigenvalues from Eq.~\eqref{eq:eigen_append} are non-negative, and thus $\Delta \succeq 0$, if and only if
	\begin{align}
		& \expect{Z_iZ_j}\leq 1+ 2t_{ij} \quad \text{if }\quad t_{ij} \in[-1, 0]\,, \nn\\
		& \expect{Z_iZ_j} \leq 1- 2 t_{ij} \quad \text{if }\quad t_{ij}\in[0, 1]\, .
	\end{align}
	Since the program from Eq.~\eqref{eq:pseudo-zizj} returns an upper bound of $\expect{Z_iZ_j}$ for any pair of qubits $(i, j)$, the Lemma follows.
	This ends the proof.
	\end{proof}
	
	\section{
		$2/3$-approximation ratio (non-constructive)}\label{sec:2_3}
	
	In this section, we show that there exists a product state $\varrho^{\rm i}$ with energy at most $-\frac{1}{4}\sum_{(ij)\in E}\omega_{ij}(1-\frac{1}{2}\expect{X_iX_j + Y_iY_j})$. This allows us to define an algorithm with a conditional approximation ratio of $2/3$. However, the proof is non-constructive and does not identify the state $\varrho^{\rm i}$.
	\begin{observation}\label{theo:non-const_ratio}
		Let $\varrho^{\rm r}$ be the ground state of the Rydberg Hamiltonian $H^{\rm r}$. The value
		\begin{equation*}
			{\rm arg}\min\bigg\{\tr(H^{\rm qmc}\varrho^{\rm r}), \sum_{(ij)\in E} \omega_{ij} \Big(1 - \frac{1}{2}\tr(X_iX_j + Y_iY_j)\varrho^{\rm r}\Big)\bigg\}
		\end{equation*}
		has a conditional approximation ratio of $\frac{2}{3}$ for QMC.
	\end{observation}
	\begin{proof}
		
		Consider the Ising Hamiltonian associated with the graph $G = (V, E, \omega)$ that is associated to a QMC instance,
		\begin{equation}
			H^{\rm ising} = \sum_{(ij)\in E}\omega_{ij} X_iX_j \, .
		\end{equation}
		Recall that the states $\ket{\pm} = \tfrac{1}{\sqrt{2}}
		\big(\ket{0} \pm \ket{1}\big)$ satisfy $X\ket{\pm} = \pm 1 \cdot \ket{\pm}$.
		Therefore, the Hamiltonian $H^{\rm ising}$ is diagonal in the product basis
		\begin{equation}
			B = \Big\{\bigotimes_{i=1}^{n} \ket{\psi_i} \ \big| \ \ket{\psi_i}\in \big\{\ket{+}, \ket{-}\big\}\Big\}\, .
		\end{equation}

		Since $B$ forms a complete basis of eigenstates, at least one of these states is a ground state of $H^{\rm ising}$. By construction, any ground state of the form $\varrho^{\rm i} = \ket{\psi^{\rm i}}\bra{\psi^{\rm i}}$ with $\ket{\psi^{\rm i}} \in B$ satisfies
		\begin{equation}\label{eq:cond_prod_ising}
			\tr(Y_iY_j \varrho^{\rm i}) = \tr(Z_iZ_j \varrho^{\rm i})  = 0 \, .
		\end{equation}

		 Next, consider a ground state $\varrho^{\rm r}$ of $H^{\rm r}$ which is invariant under 
		 \begin{equation}
		 	X_i \mapsto Y_i \quad \quad Y_i \mapsto -X_i \quad \quad Z_i \mapsto Z_i\,.
		 \end{equation}
		 It follows that
		\begin{align}
			\tr(X_iX_j \varrho^{\rm r}) = \frac{1}{2} \tr\big((X_iX_j + Y_iY_j) \varrho^{\rm r}\big)
		\end{align}
		The energy of this state under the Hamiltonian $H^{\rm ising}$ is thus given by
		\begin{equation}
			\tr(H^{\rm ising} \varrho^{\rm r}) = \frac{1}{2}\sum_{(ij)\in E}\omega_{ij} \tr((X_iX_j + Y_iY_j) \varrho^{\rm r})\, .
		\end{equation}
		However, by definition $\varrho^{\rm i}$ minimizes the energy of $H^{\rm ising}$ and thus
		\begin{align}\label{eq:ineq_exists}
			\tr(H^{\rm ising}\varrho^{\rm i}) &\leq  \tr(H^{\rm ising}\varrho^{\rm r})\nn\\
			&=\frac{1}{2} \sum_{(ij)\in E}\omega_{ij}\tr((X_iX_j + Y_iY_j) \varrho^{\rm r}) \, .
		\end{align}
		We can use the state $\varrho^{\rm i}$ to upper bound the energy of QMC. Following the condition from Eq.~\eqref{eq:cond_prod_ising}:
 		\begin{align}
			\tr(H^{\rm qmc}\varrho^{\rm i}) &= -\frac{1}{4}\sum_{(ij)\in E}\omega_{ij}\Big(1-\tr\big((X_iX_j-Y_iY_j-Z_iZ_j)\varrho^{\rm i}\big)\Big)\nn\\
			&= -\frac{1}{4}\sum_{(ij)\in E} \omega_{ij} (1-\tr(X_iX_j \varrho^{\rm i}))
		\end{align}
		Next, we write this expectation value using the bound from Eq.~\eqref{eq:ineq_exists}
		\begin{align}\label{eq:QMC_energy_prod_ideal}
			\tr(H^{\rm qmc} \varrho^{\rm i}) &= -\frac{1}{4}\sum_{(ij)\in E} \omega_{ij} (1-\tr(X_iX_j \varrho^{\rm i}))\nn\\
			&\leq -\frac{1}{4}\sum_{(ij)\in E}\omega_{ij}\Big(1 - \frac{1}{2} \tr\big((X_iX_j + Y_iY_j)\varrho^{\rm r}\big) \Big) \, .
		\end{align}
		
		To find the conditional approximation ratio of the algorithm, we substitute the upper and lower QMC bounds from Eqs.~\eqref{eq:bound},~\eqref{eq:energy_Rydberg},~\eqref{eq:QMC_energy_prod_ideal} into Eq.~\eqref{eq:lower_ratio}. This yields the conditional approximation ratio
		\begin{align}
			\kappa &\geq \max_{p\in[0, 1]}\min_{t\in [-1, \frac{1}{3})} \Bigg\{p\frac{\,  - 4t}{1-3t} + (1-p) \frac{(1- t)}{1-3t}\Bigg\}\nn\\
			&= \frac{2}{3} \, ,
		\end{align}
		for $p^*=\frac{1}{3}, t^* = -\frac{1}{3}$.
	\end{proof}
	\textit{Remark:} This algorithm only employs measurements of a Rydberg array. However, the method is non-constructive in the sense that it does not return the state which produces the approximation. 
	
\end{document}